\documentclass[letterpaper]{article}

\title{Error-Tolerant Exact Query Learning of Finite Set Partitions with Same-Cluster Oracle}
\date{\today}
\author{%
    Adela Frances DePavia\\
    \texttt{adepavia@uchicago.edu}\\
    University of Chicago
    \and
    {Olga Medrano Mart\'{i}n del Campo} \\ \texttt{omedranomdelc@uchicago.edu}\\
     University of Chicago
    \and
     {Erasmo Tani} \\ \texttt{etani@uchicago.edu}\\
     University of Chicago
}
\usepackage{arxiv_style}
\usepackage{fullpage}

\begin{document}
\maketitle

\begin{abstract}
This paper initiates the study of active learning for exact recovery of partitions exclusively through access to a same-cluster oracle in the presence of bounded adversarial error. We first highlight a novel connection between learning partitions and correlation clustering. Then we use this connection to build a R\'enyi-Ulam style analytical framework for this problem, and prove upper and lower bounds on its worst-case query complexity. Further, we bound the expected performance of a relevant randomized algorithm. Finally, we study the relationship between adaptivity and query complexity for this problem and related variants.
\end{abstract}

\section{Introduction}\label{section:introduction}

Learning cluster structure from data is a task %
with many applications. While the statistical setting is typically concerned with using batch data to approximately recover cluster structure with high probability, some applications allow for the learner to make explicit queries, and some require exact recovery guarantees. In particular, consider applications in which one has to carry out potentially expensive measurements or time-consuming experiments to learn whether two objects are part of the same class. Such settings arise in several scientific domains, particularly in bioinformatics \cite{ balding1996comparative, bouvel2005combinatorial,grebinski1998reconstructing,sorokin1996new}.

One model for such clustering tasks is that in which the learner's only information comes from having access to a \emph{same-cluster} oracle. In this setting, the learner obtains information about a hidden partition $\cC$ of a finite set $V$ by repeatedly choosing two elements $u,v \in V$ and asking questions of the form ``Are $u$ and $v$ part of the same cluster?'', i.e. given $\cC$ and $V$, a same-cluster oracle is an oracle $\alpha_{\cC}$ that takes as input $u$ and $v$ in $V$, and returns $1$ if $u$ and $v$ belong to the same cluster in $\cC$, and $-1$ otherwise. The power and limitations of same-cluster oracles have been studied in a variety of settings, including the clustering problem described above \cite{ashtiani2016clustering, LM2021, ReyzinSrivastava2007,saha2019correlation}.

However, this model makes the arguably unrealistic assumption that all queries return the correct answer. In the motivating applications, queries are often at risk of failure or subject to experimental noise, and the learner's goal is to recover the partition in spite of these errors. These errors are often not persistent: in the presence of noise, repeating an experiment multiple times may yield different answers. In spite of these errors, practitioners may need to achieve exact recovery of the underlying partition. However, existing theory either focuses on the error-free regime \cite{LM2021, ReyzinSrivastava2007} or assumes that errors are persistent and focuses on probabilistic guarantees, rather than exact recovery \cite{del2022clustering, peng2021towards}.

In this paper, we address this literature gap by formalizing and initiating the study of learning partitions via a same-cluster oracle in the presence of non-persistent errors. A typical way to incorporate uncertainty and noise is to assume that errors occur independently at random with some small probability on every query. However, it is not possible to guarantee exact recovery of the underlying clusters in this model. Instead, we focus on a model in which the errors are chosen potentially adversarially, but are bounded, so that the partition can be recovered exactly. In particular, we say a same-cluster oracle $\alpha_{\cC}$ is $\ell$-faulty if it may return an incorrect answer up to $\ell$ times. We do not assume that errors are persistent: if a single pair of elements is queried repeatedly, an $\ell$-faulty oracle may return inconsistent responses. We define the $\ell$-\emph{bounded error partition learning} (\ourproblem) problem, as the problem of exactly recovering a hidden partition via access to an $\ell$-faulty same-cluster oracle.

We introduce a two-player game based on \emph{correlation clustering}, and we show that the minimax value of this game is closely linked to the query complexity of the \ourproblem{} problem. We then use this game as a framework to give a lower bound for the query complexity of this task. We also give a simple algorithm to solve this problem and we analyze its worst-case as well as its expected query complexity.

We then explore the non-adaptive setting, in which algorithms are constrained to submit all of their queries in a single round. When considering the cost of a clustering algorithm, practitioners may seek to minimize the total number of oracle queries (i.e. experiments) performed, or they may care more about an algorithm's \textit{adaptivity}. An \textit{adaptive} algorithm can submit queries sequentially and make decisions about future queries in a manner dependent on the result of the previous queries, whereas a \textit{non-adaptive} algorithm is one which submits its full set of queries simultaneously and then makes decisions based on the full set of responses. For example, in biomedical applications massive parallel sequencing allows practitioners to submit a large number of queries in a single round; in such cases a practitioner may then be willing to submit more total queries in exchange for only having to submit a single round of queries. We leverage existing results to completely characterize the complexity of this problem in the non-adaptive same-cluster oracle model and we give general information-theoretic lower bounds for learning partitions via any faulty binary oracle.

We find that the \ourproblem{} problem occupies a unique position at the intersection of different research streams. In fact, the study of this problem complements work on clustering with same-cluster query advice, and the techniques used for its analysis draw from the theory of graph learning with oracle queries and the study of R\'enyi-Ulam liar games\footnote{We provide the reader with relevant background from these areas in Section \ref{subsection:background}.}. We believe that this problem could provide fertile ground for investigation by members of any of these research communities, and we are hopeful that the impact of its introduction will go beyond the results presented in this paper.

\subsection{Background and Related Results}\label{subsection:background}
\paragraph{Clustering with same-cluster oracles}
In the error-free regime, \citet{LM2021} studied the query complexity of recovering and verifying partitions of a finite set in different oracle models. They prove tight lower-bounds on the query complexity of learning partitions with error-free same-cluster oracles, and point out that an algorithm proposed by \citet{ReyzinSrivastava2007} exactly achieves this query complexity.

\citet{mazumdar2017clustering} initiate the study of clustering with same-cluster oracles in the presence of persistent i.i.d. errors. They point out that under these assumptions, learning a partition has a strong relationship to recovering community structure in the stochastic block model (SBM). %
This work inspired a productive line of research on the i.i.d. noise model, which also leverages connections to the SBM to study related problems. In fact, this problem has been studied in the setting when $k=2$ \cite{green2020clustering}, when the underlying clusters are nearly-balanced \cite{peng2021towards}, and in the semi-random noise model, in which errors occur with some i.i.d. probability, but when they do occur the (erroneous) answer may be chosen adversarially~\cite{del2022clustering}.

\paragraph{Active learning of graph with oracles} The field of graph learning focuses on learning graph structure with oracle access to the graph in question\footnote{In this paper, we invoke the phrase ``graph learning'' in reference to problems which seek to exactly recover graph structure by actively making queries to some graph oracle, as opposed to the wide field of problems pertaining to learning graphs from data, which often focus on approximate or probabilistic results (i.e. such topics as surveyed by \cite{Xia_2021}).}. Numerous settings have been considered in the literature, including instances where the target graphs are matchings \cite{alon2004learning}, stars or cliques \cite{alon2005learning}, or Hamiltonian cycles \cite{sorokin1996new}. Further research examines Las Vegas algorithms \cite{abasi2019learning} and the role that adaptivity plays for these problems \cite{chang2014learning, grebinski2000optimal}. In a related paper, \cite{ReyzinSrivastava2007} give an algorithm that learns the partitions of a graph with $n$ vertices using $O(nk)$ shortest path queries, as well as an algorithm that learns such partition using $O(n\log n)$ edge counting queries. \cite{angluin2008learning} give an algorithm that can learn a graph with $n$ vertices using $O(\log n)$ edge detection queries per edge. Following this line of inquiry, \cite{LM2021} find that the exact worst-case number of queries needed to learn a partition of $[n]$ into $k$ clusters in the same-cluster query model is exactly $n(k-1) - \binom{k}{2}$ if $k$ is known by the learner a priori, and $nk - \binom{k+1}{2}$ otherwise.

\paragraph{Geometric clustering with oracle advice} 
 \citet{ashtiani2016clustering} introduce the problem of clustering data when the learner can make limited same-cluster queries, and they show that this leads to a polynomial-time algorithm for k-means clustering under certain geometric assumptions. In recent years, the model of Ashtiani~\etal has inspired a variety of other related works which aim to understand the benefits of accessing a same-cluster oracle (e.g. \cite{bressan2020exact,bressan2021margin,del2022clustering} ). \citet{ailon2018approximate} show that making a small number of queries to a same-cluster oracle allows one to improve exponentially the dependency on the approximation parameter $\varepsilon$ and on the number $k$ of clusters for polynomial time approximation schemes for some versions of correlation clustering.
\citet{saha2019correlation} also consider the problem of using same-cluster queries to aid the solution of a correlation clustering instance, and they give algorithms with performance guarantees that are functions of the optimal value of the input instance.

\paragraph{Correlation clustering} 
\citet{bansal2002correlation} introduced the problem of correlation clustering. In this problem, one is given a signed graph $G=(V,E,\sigma)$, where $\sigma:E\to\{\pm1\}$ is a function representing one's prior belief about which pairs of vertices belong to the same cluster. The goal of the problem is to return a partition of the vertices of $G$ into clusters that minimizes the amount of disagreement with the edge signs given by $\sigma$ (See Section \ref{section:RU-correlation-clustering-game} for more details). The problem is known to be NP-Hard. A multitude of variations to this problem have been proposed, including versions with weighted edges \cite{ ailon2008aggregating,demaine2006correlation}, a version where the number of clusters is constrained to some fixed $k$ \cite{giotis2006correlation}, and an agreement maximization setting. Different assumptions have also been studied, such as instance stability and noisy partial information \cite{makarychev2014bilu,makarychev2015correlation,mathieu2010correlation}.

\paragraph{R\'enyi-Ulam games} In his autobiography, Stanisław Ulam, introduced the following two player game \cite{ulam1991adventures}. One player (the responder) thinks of a number in $[N]$ for some $N \in \N$, and the other player (the questioner), given $N$, tries to guess the number by asking only yes/no questions. The main twist to this setup is that the responder is allowed to lie up to $\ell$ times. We refer to this game as the \emph{Ulam's liar game}. We give a short overview of the history of this problem in Appendix~\ref{appendix:renyi-ulam}. Following the study of this question, the term ``R\'enyi-Ulam games'' has been used to identify a wide range of problems involving asking questions to an oracle who is allowed some limited amount of lying (see e.g. \cite{pelc2002searching}). The question of finding the worst-case query complexity of the \ourproblem{} problem naturally lends itself to be formulated as R\'enyi-Ulam game, in which the learning algorithm takes the role of the questioner, and the oracle plays the part of
the responder.

\subsection{Our Results}\label{subsection:our-results}
We first consider the \ourproblem{} problem in the adaptive setting. For a target partition $\cC$, let $k$ denote the number of clusters in $\cC$. There are two particular regimes of interest for the \ourproblem{} problem: that in which $k$ is known to the algorithm, and that in which $k$ is unknown. We show that the worst-case query complexity of this problem can be analyzed by considering the optimal value of a game introduced in Section~\ref{section:RU-correlation-clustering-game}, which we refer to as R\'enyi-Ulam Correlation Clustering (RUCC) game. In Section~\ref{ssec:adaptive-error-tolerant-lowerbound} we use the framework of this game to establish the following result in the nontrivial regime of $k< n$:

\begin{theorem}\label{thm:lower-bound} Whenever $k < n$, every algorithm to solve the \ourproblem{} problem must make at least:
\begin{equation}\label{eq:lower-bound-equation}
    n(k-1) - {k \choose 2} +\max \left\{{(\ell-1)\over 2}n + {k \over 2}, 0\right\} + \ell.
\end{equation}
many queries in the worst case, regardless of whether $k$ is known or unknown.
\end{theorem}
We note that the above result is independent of computational resources but rather speaks to the inherent information bottleneck of the problem. The authors consider the formulation of the RUCC game and its use as an analytical tool in establishing the above bound to be one of the key technical contributions of this work.

In Section~\ref{section:adaptativity-considerations} we study the problem under the adaptivity lens. Non-adaptive algorithms are often desirable, or even necessary, in applications where queries can be submitted in parallel but take a long time to answer. It is easy to see that the number of queries needed to solve any problem non-adaptively is always at least as large as the number of queries needed to solve the same problem adaptively. We first analyze the error-free version of the problem in the non-adaptive setting. Our results, which incorporate known results in the literature are summarized in Theorem~\ref{thm:error-free-adaptive-summary}.

We then study the impact of adaptivity in the setting of $\ell$-faulty oracles that may deliver non-persistent errors. One advantage of considering the non-adaptive regime is that it allows for different error models, including those which are not well-defined in the adaptive setting. In particular, we broaden our problem instance to include algorithms with access to arbitrary, potentially more powerful, binary oracles, and we consider an error model in which the number of errors is upperbounded by a fixed fraction of the number of queries. This is a natural model motivated in applications which is distinct from the $\ell-$faulty oracle setting. For this more general setting, we give lower bounds on the query complexity of the problem based on coding theory. We summarize our results for error-tolerant partition learning in the following theorem:
\begin{theorem}\label{thm:error-tolerant-adaptive-summary}The query complexity, denoted $q^*$, of the \ourproblem{} and its generalization to arbitrary binary oracles, is outlined in Table~\ref{table:error-tolerant-adaptivity-table} below. In the same-cluster oracle setting, the number of errors is upperbounded by a fixed parameter $\ell$. In contrast, in the non-adaptive case for arbitrary oracles, the number of errors incurred by the oracle is upperbounded by $c \cdot q$ for some $c\in(0,0.5)$, where $q$ is the number of queries made the algorithm. 
\begin{table}[h]
\vskip 0.05in
\begin{center}
\begin{small}
\begin{sc}
\begin{tabular}{c||cccc}
\hline
Adaptivity & $k$ known & $k$ unknown  \\
\hline\hline
    \rule{0pt}{4ex} Non-adaptive    & See Lemma~\ref{lemma:non-adaptive-with-errors-simple-case} %
    & $(2\ell + 1){n\choose 2}$
     \\[0.5cm]
     \hline
    \rule{0pt}{4ex} \begin{tabular}[x]{@{}c@{}}{\scriptsize Non-Adaptive}\\ {\scriptsize Arbitrary binary oracle} \\ $\ell = c\cdot q$\end{tabular}   & $q^* \geq \Omega\left(\frac{n\log k}{1-H_2(c)}\right)$ & $q^* \geq \Omega\left(\frac{n\log n}{1-H_2(c)}\right)$\\[0.5cm]
    \hline
    \rule{0pt}{4ex} Adaptive & \eqref{eq:lower-bound-equation} {\scriptsize $\leq q^*  \leq (\ell+1)\left(n(k-1) - {k \choose 2}\right) + \ell $} & {\scriptsize \eqref{eq:lower-bound-equation} $\leq q^* \leq (\ell+1)\left(nk - {k+1 \choose 2}\right) + \ell$}
\end{tabular}
\end{sc}
\end{small}
\end{center}
\caption{Query complexity bounds in the presence of $\ell$ errors. From the top left to the bottom right entry, results establishing these bounds can be found in Lemma~\ref{lemma:non-adaptive-with-errors-simple-case};  Lemma~\ref{lemma:information-theoretic-lowerbounds}; Theorem ~\ref{thm:lower-bound}; Proposition~\ref{prop:adaptive-error-tolerant-upperbound}.}\label{table:error-tolerant-adaptivity-table}
\vskip -0.1in
\end{table}
\end{theorem} 
Finally, in Section~\ref{section:las-vegas-type-bounds}, we focus on deriving Las-Vegas type guarantees for the \ourproblem{} problem. We prove a bound on the expected number of queries made by the randomized algorithm given in Section~\ref{ssec:adaptive-error-free}, as shown in the following theorem.
\begin{theorem}\label{thm:las-vegas-type-bound}
    There exists an algorithm that solves the \ourproblem{} problem by making at most
    \begin{equation}
        (\ell+1) \left({n(k+1)\over 2}-k\right) + \ell.
    \end{equation}
    queries in expectation both when $k$ is known to the algorithm and when it is not.
\end{theorem}

\section{Technical Preliminaries}\label{section:technical-preliminaires}
\paragraph{Basic definitions} Given a positive integer $n$ we denote with $[n]$ the set $\{1,..., n\}$. Given any finite set $V$ we denote with $\binom{V}{2}$ the set:
\[
    \binom{V}{2} \defeq \{\{i,j\}\subseteq V \mid i\neq j\}.
\]
A graph $G=(V,E)$ is a pair containing a finite vertex set $V$ and a subset $E\subseteq \binom{V}{2}$. A \emph{$k$-coloring} of a graph $G=(V,E)$ is a function $\chi:V \to [k]$ such that $\chi(u) \neq \chi(v)$ for every $\{u,v\} \in E$. The pre-images $\chi^{-1}(a)$ for $a \in [k]$ are called the \emph{color classes} of $\chi$. A surjective $k$-coloring is a $k$-coloring such that $\chi(V) = [k]$, i.e. one in which all color classes are nonempty. A graph is said to be \emph{$k$-colorable} if there exists some $k$-coloring $\chi$ of $G$, and it is said to be \emph{uniquely $k$-colorable} if, in addition, the partition of the vertex set induced by the preimages of any such map $\chi$ is unique. In this paper, we will work with \emph{uniquely surjectively $k$-colorable graphs}, namely those where every surjective $k$-coloring of the graph induces a unique partition as described above. From this point onward, in the paper, we will use these two terms interchangeably as, for the nontrivial setting $k<n$ the two notions are equivalent. Further discussion of this topic can be found in Appendix \ref{appendix:surjectively-uniquely-kc}. Given a graph $G=(V,E)$ a pair $\{u,v\}$ of vertices is $k$\emph{-inseparable} if there exists no surjective $k$-coloring $\chi$ of $G$ such that $\chi(u) \neq \chi(v)$.

\paragraph{Partitions} Given a finite set $V$ a $k$-partition of $V$ is a collection of $k$ pairwise disjoint non-empty subsets $\cC = \{C_1, ... , C_k\}$ of $V$ such that $\bigcup_{a=1}^k C_a = V$. We will denote by $n$ the cardinality $|V|$ of $V$ and we will often assume without loss of generality that $V = [n]$. Given two elements $u$ and $v$ of $V$, we write  $u\sim_\cC v$ (resp. $u\not\sim_{\cC} v$) for the statement $\exists C_a \in \cC, \{u,v\}\subseteq C_a$ (resp.$\nexists\; C_a \in \cC, \{u,v\}\subseteq C_a$). We denote by $B_n$ the $n^{th}$ Bell number, i.e. the number of partitions of a set of $n$ elements. The Bell numbers satisfy:
\begin{equation}\label{eq:asymptotics-of-bell-numbers}
    \log B_n = \Omega(n \log n) \hspace{1cm}(\text{\cite{de1981asymptotic})}.
\end{equation}
We denote by ${n \brace k}$ the Stirling number of the second kind, i.e. the number of ways to partition a set of $n$ labelled objects into $k$ non-empty unlabelled subsets. The Stirling numbers of the second kind satisfy:
\begin{equation}\label{eq:asymptotics-of-stirling-numbers}
    \log {n\brace k} = \Omega(n \log k) \hspace{1cm}(\text{\cite{rennie1969stirling}}).
\end{equation}

\paragraph{Coding theory background} A binary code of length $q$ is a subset $C \subseteq \{0,1\}^q$. Elements of $C$ are called \emph{codewords}. The \emph{distance} of a code $C$ is defined to be the minimum Hamming distance between any two of its codewords: $\Delta C \defeq \min_{c_1 \neq c_2 \in C} |\{ i \in [q] \mid c_1(i) \neq c_2(i)\}|$. The \emph{size} of the code is the cardinality of $C$. The \emph{relative distance} of a code $C$ with distance $\Delta C$ is defined as $\delta(C) \defeq \Delta C / q$. The \emph{rate $R(C)$} of a code is given by
$R(C) \defeq \log(\abs{C})/q$. A classical result from coding theory states that a code $C$ with relative distance $\delta(C) \geq \delta$ and rate $R(C)$ must satisfy
\begin{equation}\label{eq:bound-on-rate-of-codes}
    R(C) \leq 1-H_2(\delta/2) + o(1),
\end{equation}
where $H_2(\cdot)$ denotes the binary entropy function. A code is said to be \emph{$e$-error-correcting} if $\Delta C > 2e$. (See e.g. \cite{GuruswamiRudraSudan_codingtheory}).

\section{The R\'enyi-Ulam Correlation Clustering Game.}\label{section:RU-correlation-clustering-game} 
Recall the definition of the \ourproblem{} problem: given $V = [n]$ and an $\ell$-faulty $\alpha_{\cC}$-oracle for some partition $\cC$ of $[n]$, we wish to exactly recover the structure of $\cC$. The query complexity of this problem is the minimum number of queries to $\alpha_\cC$ that every algorithm would need to make to recover $\cC$ in the worst case, expressed as a function of $n$ and $\ell$ and the parameter $k \defeq |\cC|$. In order to analyze the query complexity of this problem, we introduce a zero-sum game, which we refer to as the \emph{R\'enyi-Ulam Correlation Clustering (RUCC) Game}.

An instance of the extended correlation clustering\footnote{We note that this setup differs slightly from an actual instance of weighted correlation clustering, in that one is allowed to have both a positive edge and a negative edge along the same pair of vertices. From an optimization perspective, the two problems are equivalent up to some constant additive change in the objective function.} problem is specified by a weighted graph with two types of edges: positive edges and negative edges, i.e. $G=(V,E^+,E^-, w^+,w^-)$, where $w^+:E^+ \to \R_{\geq0}$, and $w^-:E^- \to \R_{\geq0}$. The problem is then to find a partition $\cC$ of $V$ that achieves the minimum disagreement with $G$ defined as follows. Given a partition $\cC$ of $V$ its \emph{cost} with respect to an instance $G$ is given by:
\[
    \cost_{G}(\cC) \defeq  \sum_{\substack{uv\in E^+\\u\not\sim_\cC v}} w^+(uv) + \sum_{\substack{uv\in E^-\\ u\sim_\cC v}}w^-(uv).
\]
The cost of an instance $G$ in the extended correlation clustering problem is given by:
\[
    \textrm{CC}(G) \defeq \min_{\substack{\cC \in \cP(V(G))}} \cost_G(\cC),
\]
where $\cP(V(G))$ is the set of partitions of $V(G)$. The corresponding $k$-constrained problem minimizes the same objective over $\cP_k (V(G))$, the collection of partitions of $V$ into exactly $k$ non-empty sets:
\[
    \textrm{CC}_k(G) \defeq \min_{{\cC\in \cP_k (V(G))}} \cost_G(\cC).
\]

A graph $G$ is said to be \textit{$(\ell,k)$-consistent} if $\textrm{CC}_k(G) \leq \ell$. Similarly, $G$ is \emph{uniquely $(\ell,k)$-consistent} if the problem $\textrm{CC}_k(G)$ has a unique feasible solution $\cC$ with $\cost_G(\cC) \leq\ell$. 

For a graph $G=(V,E^+,E^-, w^+,w^-)$, we define the following weighted unsigned subgraphs: $G^+ \defeq (V, E^+, w^+)$ and $G^- \defeq (V, E^-, w^-)$.

The RUCC game is parameterized by positive integers $n$, $\ell$ and $k$, and it has two players, the questioner and the responder. At each iteration $t$ the questioner plays a pair of elements $e_t \defeq\{u,v\}\in {[n]\choose 2}$ and the responder plays a response $r_t \in \{\pm 1\}$. At any iteration $T$, we define the graph $G_T$ with vertices $V = [n]$ and edge sets $E^+_T, E^-_T$ given by
\[
    E^+_T \defeq \{e_t \mid t\in [T], r_t = 1 \}, \hspace{1cm} E^-_T \defeq \{e_t \mid t\in [T], r_t = -1 \},
\]
and weights
\[
    w^+_T (uv) \defeq |\{t\in [T] \mid e_t =\{u,v\}, r_t=1\}|, \hspace{0.4cm} w^-_T (uv) \defeq |\{t\in [T] \mid e_t =\{u,v\}, r_t = -1\}|.
\]

We enforce the following constraint: the responder's answer must be such that the graph $G_T$ is $(\ell,k)$-consistent at every iteration $T$. The game ends when $G_T$ is uniquely $(\ell,k)$-consistent. The payoff of the game to the responder is the number of iterations needed for the game to end. We refer to the \emph{value} of the game as the minimax payoff to the responder, i.e. the maximum payoff achieved by any responder strategy against the questioner's corresponding value-minimizing strategy. We summarize the main insight of this section in the following remark.
\begin{remark}
    The value of the R\'enyi-Ulam Correlation Clustering game with parameters $n, \ell$ and $k$ is the query complexity of the \ourproblem{} problem for $k$ known. We denote this value by $\cV(n,\ell,k)$.
\end{remark}

\subsection{Lower bound on $\cV(n,\ell,k)$}\label{ssec:adaptive-error-tolerant-lowerbound}
In this section, we prove Theorem~\ref{thm:lower-bound}. We establish \eqref{eq:lower-bound-equation} by demonstrating a lower bound to the value of the RUCC game by analyzing the performance of a specific response strategy. We do this by extending a technique first introduced by Liu and Mukherjee~\cite{LM2021} to the case of a faulty oracle. The essence of their strategy for the response player in the case of $\ell=0$ is for the responder to answer $-1$ for as long as possible without making the minimum cost solution for the instance greater than $0$, i.e. while maintaining $(0,k)$-consistency. 

The main strategy we will adopt for the response player is the following: the player will answer $1$ when a pair $\{u,v\}$ at time $T$ is queried if and only if the pair is $k$-inseparable in $G_T^-$. Intuitively, this corresponds to the response player answering $-1$ for as long as possible without ever raising the optimal value of the correlation clustering instance above $0$.

Analyzing this game yields following lemma, which directly implies Theorem~\ref{thm:lower-bound}: 
\begin{lemma}\label{lemma:LM-responder-lowerbound}
    Whenever $k < n$:
    \begin{equation}\label{eq:RUCC}
        V(n,\ell, k) \geq n(k-1) - {k \choose 2} +\max \left\{{(\ell-1)\over 2}n + {k \over 2}, 0\right\} + \ell.
    \end{equation}
\end{lemma}
\begin{proof}
    We first describe a responder strategy that achieves a bound of:
    \begin{equation}\label{eq:RUCC-noell}
        n(k-1) - {k \choose 2} +\max \left\{{(\ell-1)\over 2}n + {k \over 2}, 0\right\},
    \end{equation}
    and then we show how to modify this strategy so as to achieve the extra $\ell$ term.
    
    Consider the responder strategy that responds with $-1$ whenever possible, while maintaining the property that $G_T$ is $(0,k)$-consistent. We now show that for any questioner strategy, the Ulam-R\'enyi Correlation Clustering game with this responder must last at least the number of iterations given in~\eqref{eq:RUCC-noell}.
    
    In the Liu \& Mukherjee responder strategy, at every iteration the graph $G_t = (V,E^+_t,E^-_t,w^+_t,w^-_t)$ is $(0,k)$-consistent by definition. In particular, this implies that when the game ends at iteration $\Tfinal$, $G_{\Tfinal}$ must be not only uniquely $(\ell,k)$-consistent, but also uniquely $(0,k)$-consistent. In particular, there is a unique $k$-partition $\cC=\{C_a\}_{a=1}^k$ of the elements of $V$ which has zero disagreement with $G_{\Tfinal}$. Let $n_a \defeq \abs{C_a}$. For any $C_a, C_b \in \cC$, let $G_{\Tfinal}[C_a \cup C_b]$ denote the subgraph of $G_{\Tfinal}$ induced on $C_a$ and $C_b$, and let $\deg_{a,b}(v)$ denote the weighted degree of a vertex in this subgraph. 
    
    Since $G_{\Tfinal}$ must be uniquely $(\ell,k)$-consistent, for any element $v \in V$, the partition $\cC'$ obtained from $\cC$ by moving $v$ from some $C_a\ni v$ to some other cluster $C_b$ should incur a cost of at least $\ell+1$. This, in particular, implies that $\deg_{a,b}(v) \geq \ell + 1 $ for every $v\in V$, and $a, b\in [k]$. (Note that this implicitly uses the fact, that every neighbor of $v$ in $G^+_{\Tfinal}$ is contained in $C_a$.)

    Liu and Mukherjee showed that in order for the graph $G_{\Tfinal}$ generated by the response strategy in the statement of Lemma~\ref{lemma:LM-responder-lowerbound} to be uniquely $(0,k)$-consistent, the graph $G^-_{\Tfinal}$ must be a uniquely $k$-colorable graph. And in particular, this implies that the total weight of the edges of $G^-_{\Tfinal}[C_a \cup C_b]$ must be at least $n_a+n_b-1$. We will assume, without loss of generality, that there exists some set $F_{a,b} \subseteq E(G^-_{\Tfinal}[C_a \cup C_b])$ which has weight exactly $n_a + n_b -1$.

    \begin{figure}
    \centering
    \hspace*{1.8cm}    
    \includegraphics[scale=0.4]{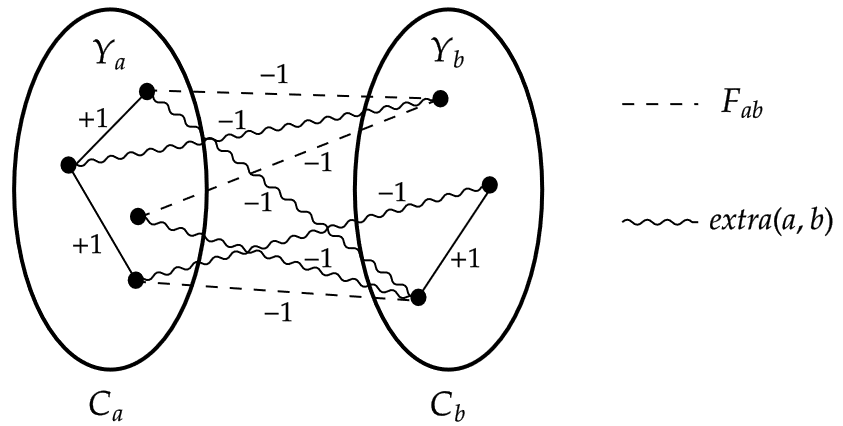}
    \caption{$G_{\Tfinal}[C_{a}\cup C_{b}]$ in the proof of Lemma \ref{lemma:LM-responder-lowerbound}. The edge sets $Y_a$ and $Y_b$ are not graphically represented, but refer to intra-cluster edges contained in $C_a$ and $C_b$.}
    \label{fig:lm-responder}
    \end{figure}
    Fix $C_a, C_b$, for $\{a,b\}\in{[k]\choose 2}$ and consider the subgraph $G_{\Tfinal}[C_a \cup C_b]$. The edges appearing in this subgraph can be divided into four parts, as Figure~\ref{fig:lm-responder} shows: the set $F_{a,b}$, the set 
    \[
        \operatorname{extra}(a,b) \defeq E(G^-_{\Tfinal}[C_a \cup C_b])\setminus F_{a,b},
    \]
    and the sets
    \[
        Y_a \defeq E(G^{+}_{\Tfinal}[C_a]) \hspace{0.5cm}\text{and}\hspace{0.5cm}
        Y_b \defeq E(G^{+}_{\Tfinal}[C_b]).
    \]
    
    The lower bound $\deg_{a,b}(v) \geq \ell + 1$ then implies that for any $a \neq b$:
    \begin{align*}
        (\ell+1) &\leq\min_{v\in C_a \cup C_b} \deg_{a,b}(v) \leq { \sum_{v \in C_a \cup C_b}\deg_{a,b}(v)\over |C_a \cup C_b|} = {2w(E(G[C_a \cup C_b]))\over n_a + n_b}
    \end{align*}
    re-arranging yields:
    \begin{align*}
        {(\ell+1)(n_a + n_b)\over 2} &\leq {w(E(G[C_a \cup C_b]))} = n_a + n_b -1 + w(Y_a) + w(Y_b) + w(\operatorname{extra}(a,b))
    \end{align*}
    which, given the non-negativity of edge-weights, implies:
    \begin{align*}
        w(\operatorname{extra}(a,b))
        &\geq \max\left\{{(\ell-1)(n_a + n_b)\over 2} +1- w(Y_a) - w(Y_b), 0 \right\}
    \end{align*}

    In particular, at the end of the last iteration $T_{\text{final}}$ the total weight of edges in $G_{T_{\text{final}}}$ is given by:
    \begin{align*}
        w(E(G_{T_{\text{final}}})) &= n(k-1) - \binom{k}{2} + \sum_{a\in [k]} w(Y_a)  + \sum_{a,b\in\binom{[k]}{2}} w(\operatorname{extra}(a,b))\\
        &\geq n(k-1) - \binom{k}{2} + \sum_{a\in [k]} w(Y_a)  +  \sum_{a,b\in\binom{[k]}{2}} \max\left\{{{(\ell-1)(n_a + n_b)\over 2} +1 - w(Y_a) - w(Y_b)
        , 0} \right\}
    \end{align*} 
    This bound holds for all valid values of $\ell$. In particular, consider the case $\ell = 0$, which corresponds to the error-free regime. For this value, we observe that the value inside the maximum term is always $0$: $(\ell -1) \leq 0$,  the partition $\cC$ of $G_{\Tfinal}$ is a $k$-partition thus $n_a \geq 1$ for all $a\in [k]$, and $w(Y_a)$ and $w(Y_b)$ are non-negative. Thus for $\ell = 0$ we recover the bound established by \cite{LM2021} in the error-free regime.
    \[
        w(E(G_{T_{\text{final}}})) \geq n(k-1) - \binom{k}{2} + \sum_{a\in [k]} w(Y_a) \geq n(k-1) - \binom{k}{2}
    \] 

    For the case that $\ell \geq 1$, we can relax the lower bound by considering continuous values in place of integer-valued $w(Y_a)$:
    \begin{align*}
        w(E(G_{T_{\text{final}}})) &\geq n(k-1) - \binom{k}{2} + \min_{y\in \R^k_{\geq 0}} \left(\sum_{a\in [k]} y_a  + \sum_{a,b\in\binom{[k]}{2}} { \max\left\{{(\ell-1)(n_a + n_b)\over 2} +1 - y_a - y_b, 0 \right\}}\right)\\
        &\geq n(k-1) - {k\choose 2} + \frac{(\ell -1 )n}{2} + \frac{k}{2}
    \end{align*}
    Because the total weight of edges in $G_{\Tfinal}$ is exactly equal to $\Tfinal$, the bound in~\eqref{eq:RUCC-noell} follows.\\

    We now show how to achieve the extra $+\ell$ term in~\eqref{eq:RUCC}, which will complete the proof of the lemma. We consider the following modification of the responder strategy above: consider the iteration at time $T^* \defeq T_{final}-1$, right before the game ends.  By then, there exist two partions $\cC_{1}$ and $\cC_{2}$ satisfying:
    \[
        \operatorname{cost}_{G_{T^*}}(\cC_{1})= \ell \hspace{1cm}\text{and}\hspace{1cm}\operatorname{cost}_{G_{T^*}}(\cC_{2})=0.
    \]
    The responder strategy described above would continue to respond in a way that agrees with $\cC_{2}$ even if that causes the game to terminate after the next query. Instead, we modify the responder strategy to start answering each of the remaining queries in a way that agrees with $\cC_{1}$, rather than $\cC_{2}$. With this strategy, the cost of $\cC_{1}$ in the correlation clustering instance is left unchanged after every query, and hence in order for the game to terminate, the cost of $\cC_{2}$ needs to rise above $\ell$. This requires at least $\ell$ additional queries, resulting in the bound given in \eqref{eq:RUCC}. 

\end{proof}
\vspace{-2mm}
\section{Adaptivity Considerations}\label{section:adaptativity-considerations}

In this section, we prove Theorems~\ref{thm:error-free-adaptive-summary} and ~\ref{thm:error-tolerant-adaptive-summary}. These results examine the relationship between query complexity and adaptivity of algorithms for active learning of partitions. We first summarize the landscape of query complexities for learning partitions in the absence of errors: in Theorem~\ref{thm:error-free-adaptive-summary}, we complement results known in literature  with our own results from Lemmas~\ref{lemma:k-known-non-adaptive-bound} and ~\ref{lemma:error-free-k-unknown} to give a complete characterization of the adaptivity-query complexity trade-off:
\begin{theorem}\label{thm:error-free-adaptive-summary}
    The query complexities of the problem of learning partitions with access to a same-cluster oracle in different regimes are given in Table~\ref{table:adaptivity-table}.
\begin{table}[h]
\vskip 0.05in
\begin{center}
\begin{small}
\begin{sc}
\begin{tabular}{c||ccc}
\hline
    \rule{0pt}{3ex} Adaptivity & $k$ known &  $k$ unknown  \\[0.1cm]
\hline\hline
     \rule{0pt}{4ex} Non-adaptive    &  See Theorem~\ref{thm:k-known-non-adaptive-bound} & ${n\choose 2}$
     \\[0.5cm]
     \hline
    \rule{0pt}{4ex} \begin{tabular}[x]{@{}c@{}}Adaptive\\(At most $k$ adaptive rounds)\end{tabular}  & $n(k-1) - {k\choose 2}$  & $nk - {k+1\choose 2}$ \\
\end{tabular}
\end{sc}
\end{small}
\caption{Adaptivity bounds in the error-free regime. The top row is discussed in Section~\ref{ssec:non-adaptive-error-free} and the bottom row is discussed in Section~\ref{ssec:adaptive-error-free}.}\label{table:adaptivity-table}
\end{center}

\end{table}
\end{theorem}

\subsection{Adaptive Algorithms for the Error-Free Regime}\label{ssec:adaptive-error-free}
In the adaptive setting with no errors, the algorithms of Reyzin and Srivastava for the case of $k$-known and $k$-unknown (Algorithm~\ref{algo:vanilla-RS-k-known} and Algorithm~\ref{algo:vanilla-RS-k-unknown}) are known to achieve the optimal number of queries for the problem \cite{LM2021,ReyzinSrivastava2007}. These algorithms can be thought of as repeatedly picking an element and inserting it into the right cluster in the partition. This is done by comparing the element to one element of each cluster already represented until a match is found. We note that these two algorithms can be made to run in $k-1$ adaptive rounds of questions and $k$ adaptive rounds of questions respectively. We rewrite these algorithms to highlight this property in Algorithm~\ref{algo:parallel-learn-k} and Algorithm~\ref{algo:parallel-learn}. 

A simple consequence of this is that, in the $k$-known (resp. $k$-unknown) regime, constraining the rounds of adaptivity to be at most $k-1$ (resp. $k$) has no impact on the number queries needed to solve the problem.

\subsection{Non-adaptive Algorithms in Error-Free Regime}\label{ssec:non-adaptive-error-free}
In this section we consider algorithms that learn a partition by making a single round of queries to an error-free same-cluster oracle. Note that one can learn the entire partition structure in a single round by querying every pair of elements, thus making $\binom{n}{2}$ many queries. For the case in which $k$ is unknown, this is actually tight, as illustrated in the following result, which is a simple consequence of Proposition 1 from \cite{ReyzinSrivastava2007}. For completeness we include a self-contained proof of this in Appendix~\ref{appendix:deferred-proofs}.
\begin{lemma}\label{lemma:error-free-k-unknown}
There exists no algorithm to learn partitions that uses a single round of queries to a same-cluster oracle, and makes fewer than $\binom{n}{2}$ queries when $k$ is unknown.
\end{lemma}

We now consider the case of non-adaptive algorithms for $k$ known. 
When $k=2$, one can learn a partition with only $n-1$ queries. Interestingly, the situation changes drastically when $k\geq 4$. In this case it is impossible to beat the $\binom{n}{2}$ bound by a nontrivial margin, while the case $k=3$ lies in an intermediate regime. Results for all values of $n, k<n$ are stated in the following theorem, which is proved in Appendix~\ref{appendix:proof-of-non-adaptive-results}.
\begin{theorem}\label{thm:k-known-non-adaptive-bound}
    For $k$ known and $k<n$, the number of same-cluster queries needed to learn a partition of $n$ items into $k$ clusters is as described in Table~\ref{table:non-adaptive-k-known-results}:
    \begin{table}[h]
    \begin{center}
    \begin{sc}
    \begin{tabular}{c||@{\hskip 0.25in}c @{\hskip 0.5in} c @{\hskip 0.5in} c @{\hskip 0.5in} c}
        \rule{0pt}{4ex} Values of $(k,n)$ & $k=2$ & \begin{tabular}[x]{@{}c@{}}$k = 3$\\$n = 4$\end{tabular} & \begin{tabular}[x]{@{}c@{}}$k = 3$\\$n \geq 5$\end{tabular} & $k\geq 4$\\
    \hline
        \rule{0pt}{4ex} Queries needed & $n-1$ & 5 & $\binom{n}{2}-\left\lfloor\frac{n}{2}\right\rfloor$ & $\binom{n}{2}-1$
    \end{tabular}
    \end{sc}
    \caption{Query complexity of non-adaptive cluster learning in the error-free regime, for $k$ known.}\label{table:non-adaptive-k-known-results}
    \end{center}
    \end{table}

    \end{theorem}

\subsection{Adaptive Algorithms in the Error-Tolerant Regime}\label{ssec:adaptive-error-tolerant-upperbound}
In this section we give a simple randomized algorithm to solve the problem of learning partitions with bounded error. The algorithm is an instantiation of the following framework: given any algorithm for the error-free version of the problem which makes $q_{\text{error\_free}}$ queries in the worst case, repeat each query the algorithm makes until $\ell+1$ of the answers agree. It is easy to see that this framework always leads to an algorithm for the error-tolerant regime that makes exactly $(\ell+1)q_{\text{error\_free}} +\ell$ queries in the worst case. See Lemma~\ref{lemma:repetition-bound} in Appendix~\ref{appendix:deferred-proofs} for a formal statement and proof of this result. Instantiating this framework with the algorithms of \cite{ReyzinSrivastava2007} gives the following:
\begin{proposition}\label{prop:adaptive-error-tolerant-upperbound}
    There exist algorithms to learn partitions from same-cluster queries with bounded adversarial error that makes at most:
    \begin{equation}\label{eq:upper-bound-expression}
        (\ell+1)\left(n(k-1) - {k \choose 2}\right) + \ell\hspace{0.5cm}\text{and}\hspace{0.5cm}
        (\ell+1)\left(nk - {k+1 \choose 2}\right) + \ell
    \end{equation} 
    queries for the case of $k$ known and $k$ unknown, respectively.
\end{proposition}
Pseudocode for the specifics of the algorithms in the above proposition is given in Algorithm~\ref{algo:Robust-RS} and Algorithm~\ref{algo:Robust-RS-k} in Appendix~\ref{appendix:algorithms}. While the above only considers the worst-case query complexity, in Section~\ref{section:las-vegas-type-bounds} we analyze the expected number of queries these algorithms make.

\subsection{Non-Adaptive Algorithms in the Presence of Errors}\label{section:non-adaptive-algorithms}
A simple consequence of the discussion in Sections~\ref{ssec:adaptive-error-free} and ~\ref{ssec:non-adaptive-error-free} is the following tight characterization of the number of queries needed for any non-adaptive algorithm to solve \ourproblem{}:
\begin{lemma}\label{lemma:non-adaptive-with-errors-simple-case}
    The number of same-cluster queries to an $\ell$-faulty oracle that any non-adaptive algorithm must make to learn a partition of $n\geq 5$ items into $k<n$ clusters is $(2\ell + 1)(n-1)$ when $k$ is known and $k=2$; $(2\ell + 1)\left(\binom{n}{2}-\lfloor n/2\rfloor\right)$ when $k$ is known and $k = 3$; $(2\ell + 1)\left(\binom{n}{2}-1\right)$ when $k$ is known and $k\geq 4$;  $(2\ell + 1)\binom{n}{2}$ queries when $k$ is unknown.
\end{lemma}

While the problem is easily solvable for the case of same-cluster oracles, we now give an information theoretic lower bound on the number of queries for the non-adaptive version of \ourproblem{} where one replaces the same-cluster oracle with an arbitrary binary oracle on the partition $\cC$. Many such oracles can be defined depending on the application. For example, consider the following oracles $\beta,\gamma, \zeta$ and $\eta$ defined on any subsets $S, S_1, S_2 \subseteq V$:
\begin{align*}
    &\beta(S)  \defeq \begin{cases}
        1 & \exists u,v\in S : u\sim_{\mathcal{C}}v\\
        0 & \text{otherwise}
    \end{cases} \qquad\qquad \gamma(S) \defeq \begin{cases}
        1 & \exists a\in[k] : C_{a}\subseteq S \\
        0 & \text{otherwise}
    \end{cases} \\ &\zeta(S)  \defeq \begin{cases}
        1 & \forall a\in[k] : C_{a}\cap S\neq \varnothing \\
        0 & \text{otherwise}
    \end{cases}\quad\eta(S_1, S_2) \defeq \begin{cases}
        1 & \exists C_a: S_1 \cap C_a\neq \varnothing\text{ and }S_2\cap C_a\neq\varnothing \\
        0 & \text{otherwise}
    \end{cases}
\end{align*}

One natural error model to consider for the non-adaptive general oracle case is that in which $\ell$ the number of errors which may occur is some constant fraction $c$ of the total queries $q$ submitted by the algorithm (i.e. $\ell = c \cdot q$). For this setting, our bounds obtained in Section~\ref{ssec:adaptive-error-tolerant-lowerbound} and~\ref{ssec:adaptive-error-tolerant-upperbound} do not apply, however one can still apply information theoretic techniques to obtain meaningful lower bounds.

In the error-free regime, such information theoretic lower bounds on query-complexity for similar problems are well known \cite{angluin2008learning, LM2021, ReyzinSrivastava2007}. In particular, \cite{ReyzinSrivastava2007} give a lower-bound on learning partitions of $\Omega(n \log k)$ for $k$-known, $\Omega(n\log n)$ for $k$-unknown. We adapt these arguments to the error-tolerant regime and thus obtain new lower bounds for the setting $\ell = c\cdot q$.

\begin{lemma}\label{lemma:information-theoretic-lowerbounds}
    Consider an arbitrary $\ell$-faulty binary oracle on partition $\cC$. In the regime in which up to a $c$ fraction of the responses to the queries can be errors, any non-adaptive algorithm learning a hidden partition in this setting must submit at least 
    \[
        \Omega\left(\frac{n\log n}{1-H_2(c)}\right) \hspace{1cm} \text{ and } \hspace{1cm} \Omega\left(\frac{n\log k}{1-H_2(c)}\right) 
    \]
    many queries for the case of $k$ unknown and $k$ known, respectively.
\end{lemma}
Proof of this result is deferred to Appendix~\ref{appendix:deferred-proofs}.

\section{Las Vegas-Type Bounds}\label{section:las-vegas-type-bounds}
In this section, we bound the expected number of queries for a simple randomized strategy to learn partitions in the presence of errors. The central result of this section is that Algorithm~\ref{algo:Robust-RS} exactly recovers a partition $\cC$ making:
\begin{equation}
    (\ell+1) \left({n(k+1)\over 2}-k\right) + \ell
\end{equation}
queries in expectation using an $\ell$-faulty same-cluster oracle. This bound applies for both the case of $k$ known and $k$ unknown.

We first consider the following variant of the algorithms of Reyzin and Srivastava for the error-free regime. We think of these algorithms as processing a single vertex at a time, by iteratively looking for a pre-existing cluster this vertex might belong to, and assigning the vertex to a newly created cluster if none is found. We now analyze the performance of these algorithms when the ordering in which the vertices are processed is given by a permutation chosen uniformly at random. We refer to these algorithms as \texttt{Randomized\_RS} in the case of $k$ unknown, and \texttt{Randomized\_RS\_k} for the case of $k$ known. (See Algorithm~\ref{algo:randomized-k-unknown} and Algorithm~\ref{algo:randomized-k-known} respectively, in Appendix~\ref{appendix:algorithms}). Fix some partition $\cC$ and let $\alpha_{\cC}$ be an error-free same-cluster oracle for $\cC$. Let random variables $Q_{\cC}$ 
 and $Q_{\cC}^{k}$ denote the number of queries performed by \texttt{Randomized\_RS} and \texttt{Randomized\_RS\_k} respectively when given access to $\alpha_{\cC}$. We note that by a simple coupling argument, $\E{}{Q_{\cC}^{k}} \leq \E{}{Q_{\cC}}$. We now give an exact expression and upper bound on $\E{}{Q_{\cC}}$ and $\E{}{Q^{k}_{\cC}}$. Proof of this result can be found in Appendix~\ref{appendix:deferred-proofs}. 
\begin{lemma}\label{lemma:RS-expected-queries}
    Let $\cC =\{C_1, ... , C_k\}$ with $|C_a|=n_a$. Then:
    \begin{equation}\label{eq:general-expected-number-of-queries-las-vegas}
        \E{}{Q_{\cC}}= (n-k) + \sum_{\substack{a,b \in[k] \\a \neq b}} {n_an_b \over n_a + n_b},
    \end{equation}
    and in particular, for any partition $\cC$
    \begin{equation}\label{eq:all-equal-bound}
        \E{}{Q^{k}_{\cC}}\leq \E{}{Q_{\cC}} \leq {n(k+1)\over 2}-k.
    \end{equation}
\end{lemma}
In particular, Equation~\ref{eq:all-equal-bound} shows that the worst-case expected number of queries is attained when the partition is completely balanced, i.e. every cluster has size approximately $n/k$.

We now consider the error-tolerant regime. Let random variables $R_{\cC}$ and $R_{\cC}^{k}$ denote the number of queries performed by Algorithms~\ref{algo:Robust-RS} and ~\ref{algo:Robust-RS-k} respectively when given access to an $\ell$-faulty oracle $\alpha_{\cC}$. By linearity of expectation, the result in Lemma~\ref{lemma:RS-expected-queries} allows us to establish the desired result:
\[
    \E{}{R_{\cC}} \leq \E{}{R_{\cC}^k} \leq (\ell+1) \left({n(k+1)\over 2}-k\right) + \ell
\]
which implies Theorem~\ref{thm:las-vegas-type-bound}.

\section{Conclusion and Future Directions}\label{section:future-directions}
In this paper we initiated the study of learning partitions from same-cluster queries subject to bounded adversarial errors. The results in this paper open up several avenues for future research. We considered the problem's worst-case query complexity in both the adaptive and the non-adaptive setting, and established a game-theoretic framework for analyzing the problem. The gap between the upper and the lower bound in the number of queries in the adaptive case is still significant (compare \eqref{eq:lower-bound-equation} and \eqref{eq:upper-bound-expression}), and we suspect that our RUCC game framework could be further leveraged to find the exact values of $V(n,\ell,k)$. 

This regime also seems to have potential connections with the theory of error-correcting codes with noiseless feedback \cite{gupta2022binary, haeupler2015communication} and more generally error-resilient interactive communication~\cite{gupta2022optimal}, and we believe it would be valuable to further explore this link. It would also be interesting to extend the results in this paper on error-tolerant learning from arbitrary oracles, such as the ones introduced in Section~\ref{section:non-adaptive-algorithms}, by providing complementary upper bounds.

\section{Acknowledgements}
The authors wish to thank Daniel Di Benedetto for providing helpful advice early on in the project; Yibo Jiang for suggesting related work; Mark Olson, Angela Wang, and Nathan Waniorek for useful discussions.

\bibliographystyle{abbrvnat}
\bibliography{library}

\appendix

\section{Proof of Theorem~\ref{thm:k-known-non-adaptive-bound}}\label{appendix:proof-of-non-adaptive-results}

In this section, we prove the results needed to complete the proof of Theorem~\ref{thm:k-known-non-adaptive-bound}. We begin by noting that every non-adaptive algorithm for learning partitions can be thought of as submitting a graph $G = (V,E)$ of queries, i.e. $E = \{uv \mid \alpha(u,v) \text{ was queried by the algorithm}\}$. We refer to this graph as the \emph{query graph}. We will also consider the complement of the query graph $\bar{G} = (V,\bar{E})$ where $\bar{E} = \{uv \mid \alpha(u,v) \text{ was not queried by the algorithm}\}$. We will refer to this graph as the \emph{unqueried graph}. 

\begin{lemma}
    Let $n>k = 2$, then the number of queries needed by a non-adaptive algorithm to learn partitions when $k$ is known is $n-1$.
\end{lemma}
\begin{proof}
It is easy to see that when $k=2$, one can learn a partition by querying any tree on $V$, thus making it possible to learn the entire partition structure non-adaptively with $n-1$ queries. In fact, the parallel Reyzin-Srivastava algorithm (discussed in Section~\ref{ssec:adaptive-error-free}) is a special case of this strategy, in which the tree queried is a star, and it recovers the optimal query bound. Note that querying fewer than $n-1$ pairs leaves the query graph disconnected making it impossible to recover the partition, hence the bound is tight.
\end{proof}

We will make use of the following result.
\begin{lemma}\label{lem:overlapping-edges-lemma}
If $n > k\geq 3$, and the unqueried graph $\bar{G}$ contains two incident edges $uv$ and $vw$, then there is a set of answers to the queries which is compatible with two distinct $k$-partitions of $V$.
\end{lemma}
\begin{proof}
Consider querying a graph $G$ which does not contain $uv$ and $vw$ for some $u, v, w\in V$. Let $V' \defeq V\setminus \{v\}$, suppose all the queries not involving $v$ are consistent with a partition of $V'$ which places $u$ in one cluster by itself, $w$ in a cluster by itself, and all the remaining vertices of $V'$ in a partition $\cC'$ containing the $k-2$ remaining clusters. Note that this makes use of the assumption that $k \geq 3$. Further, suppose that all the queries made involving $v$ returned $-1$. Then two possible partitions of $V$ are consistent with this set of queries: $\{\{v,u\},\{w\}\}\cup\cC'$ in which $\{v,u\}$ is a cluster, and $\{\{u\},\{v,w\}\}\cup\cC'$ in which $\{v,w\}$ is a cluster.
\end{proof}

\begin{lemma}\label{lemma:nonadaptative-errorfree-k-equals3}
Let $n\geq 5 $ and $k = 3$, then the number of queries needed by a non-adaptive algorithm to learn partitions when $k$ is known is $\binom{n}{2}-\left\lfloor\frac{n}{2}\right\rfloor$. \footnote{For the case of $(k,n)=(3,4)$, a simple argument shows that one needs $5=\binom{4}{2}-1$ queries to learn the correct partition.}
\end{lemma}
\begin{proof}
    Let the underlying partition be equal to $\cC = \{P_{A}, P_{B}, P_{C}\}$. We begin by showing that $\binom{n}{2}-\left\lfloor\frac{n}{2}\right\rfloor$ queries suffice to learn $\cC$. It is easy to see that the following proof also yields an efficient algorithm for recovering the partition $\cC$ given the results of the queries.
    
    Let $(U,L)$ be a balanced partition of $V$, i.e. $U \sqcup L = V$ and $||U|-|L||\leq 1$. Consider the case in which the unqueried graph $\bar{G}$ consists of an arbitrary maximum matching between the elements in $U$ and $L$ and the query graph $G$ contains all the remaining edges. Note that $|E|=\binom{n}{2}-\left\lfloor\frac{n}{2}\right\rfloor$. Because the complete graphs $G[U]$ and $G[L]$ are queried, the restrictions of $\cC$ on $U$ and on $L$ can be determined from the responses. By Pigeonhole Principle, at least one of the two restrictions consists of at least two clusters. Without loss of generality, we will assume that the restriction of $\cC$ onto $U$ contains at least two clusters.
    
    Later, one may extend the restriction of $\cC$ on $U$ to include all of $V$, mainly by using the following principle, which leverages the assumption $k=3$: \emph{if there is a subset $S\subset V$ such that the restriction of $\cC$ on $S$ is uniquely determined, and $u,w\in S$ belong to distinct clusters, then given $v\not\in S$, if $vu, vw \in E$, then the restriction of $\cC$ to $S\cup\{v\}$ can be uniquely determined}. 
    
    We have the following cases, without loss of generality:
    \begin{itemize}
        \item \textbf{Case 1. }There exist three vertices in $U$, such that each is in a different cluster. Every vertex $v\in L$ is queried together with at least two of these three vertices, so its cluster is determined by the principle above. %
        \item \textbf{Case 2. }All vertices in $U$ lie in $P_{A}$ and $P_{B}$. By the above principle we can determine the cluster of all but at most one vertex $v\in L$, which would only remain if $vw\in\bar{E}$, where $w$ is the only $w\in U\cap P_{A}$. If 
        all the other vertices in $L$ lie in $P_{A}\cup P_{B}$, then $v\in P_{C}$. Otherwise, there is a vertex $u\in L\cap P_{C}$ and by the above principle, the cluster of $v$ can be determined. 

    \end{itemize}

    We now show that any algorithm making strictly fewer than $\binom{n}{2}-\left\lfloor\frac{n}{2}\right\rfloor$ queries would not be able to recover the correct partition in general. In particular, we show that for any set of fewer than $\binom{n}{2}-\left\lfloor\frac{n}{2}\right\rfloor$ queries, there is a set of answers which are compatible with more than one 3-partition of $V$. We begin by noting that if the algorithm omits more than $\left\lfloor\frac{n}{2}\right\rfloor$ then the unqueried graph $\bar{G}$ (defined above) cannot be a matching. Hence, $\bar{G}$ must contain two incident edges $uv$ and $vw$ and the algorithm will not be able to guarantee recovery of the partition, by Lemma~\ref{lem:overlapping-edges-lemma}.
\end{proof}

\begin{lemma}\label{lemma:k-known-non-adaptive-bound}
Let $n>k \geq 4$, then the number of queries needed by a non-adaptive algorithm to learn partitions when $k$ is known is $\binom{n}{2}-1$.
\end{lemma}
\begin{proof}[\textbf{Proof of Lemma~\ref{lemma:k-known-non-adaptive-bound}}]
    We begin by showing that a non-adaptive algorithm querying all but one pair $u,v$ of elements of $V$ can efficiently recover the underlying partition $\cC$ independently of the answers received. Note that this requires exactly $\binom{n}{2}-1$ queries. We describe the algorithm at a high level, but it is easy to see that all the operations executed by the algorithm can be made to run in at most $O(n^2)$ time.
    
    First, consider the elements of $V'\defeq V \setminus\{u,v\}$. Given that every pair in this set was queried, upon receiving the answer to the $\binom{n}{2}-1$ queries, the algorithm can use the information received to recover the restriction of the partition $\cC$ to $V'$, and label all the elements of $V'$ accordingly. This can be achieved by simply running breadth-first search on the graph composed of all the 1 answers on $V'$. At this point the algorithm has labelled the elements of $V'$ in a way that is consistent with $\cC$, and it knows how many distinct clusters of $\cC$ are represented in $V'$. Because $\cC$ comprises of $k$ clusters on $V$, and only $2$ nodes are unaccounted for in $V'$ then the restriction of $\cC$ on $V'$ must comprise of either $k$, $k-1$, or $k-2$ distinct clusters.
        
    If this restriction comprises of $k$ distinct clusters, then all the clusters in the partition are represented in $V'$. In this case, there must exist (not necessarily distinct) elements $w_v, w_u \in V'$ such that $w_v \sim_{\cC} v$ and $w_u \sim_{\cC} u$ and since the pairs $(w_v,v)$ and $(w_u,u)$ have been queried by the algorithm, the algorithm can recover the correct partition structure by simply assigning $u$ and $v$ to the same clusters as $w_u$ and $w_v$ respectively.

    On the other hand, if the restriction  of $\cC$ on $V'$ comprises of $k-1$ clusters, then either $u$ and $v$ make up the last cluster; or $u$ alone makes up the last cluster and $v$ is part of a cluster already represented in $V'$; or vice-versa. The algorithm can distinguish between these three cases by checking which of $u$ or $v$, if either, belong to the same cluster as some $w\in V'$, something which can be easily determined from the queries made. This resolves the full structure of $\cC$ on $V$.
    
    Finally, if the restriction of $\cC$ on $V'$ contains $k-2$ distinct clusters, then $u$ and $v$ must each make up their own separate cluster, and hence the correct partition $\cC$ is simply the partition of $V'$ learned by the algorithm with two extra clusters: $\{u\}$ and $\{v\}$.

    This completes the proof that $\binom{n}{2}-1$ queries suffice to learn the partition.

    We now show that $\binom{n}{2}-1$ queries are in fact necessary to learn $\cC$ in the worst case. In other words, for any set of at most ${n\choose 2} - 2$ queries, there exist a set of answers to these queries which is consistent with multiple distinct partitions. We consider two cases: one in which there exist two omitted queries which are incident (i.e. they overlap by an element), and one in which all omitted queries are disjoint.
    
    In the first case, the result holds by Lemma~\ref{lem:overlapping-edges-lemma}.

    We then consider the case when the two omitted queries $\{u, v\}$ and $\{w, x\}$ are disjoint. In this case, there exists a subset $V' = V\setminus \{u, v, w, x\}$ such that all pairs in $V'$ have been queried, and all pairs $\{a, b\}$ have been queried for $a \in \{u, v, w, x\}$ and $b\in V'$. Because all pairs in $V'$ have been queried, we can fully recover the restriction of $\cC$ on $V'$. Suppose the restriction of $\cC$ on $V'$ comprises $k-3$ clusters and that all queries made involving $u, v, w$ and $x$ return $-1$. Then it is impossible to distinguish between the cases when the two remaining clusters are $\left\{\{u,v\},\{w\},\{x\} \right\}$ or $\left\{\{u\},\{v\},\{w, x\} \right\}$. Note that this last part crucially relies on the fact that $k > 3$.
\end{proof}

\section{Deferred Proofs}\label{appendix:deferred-proofs}
\begin{proof}[\textbf{Proof of Lemma~\ref{lemma:error-free-k-unknown}.}]
\begin{figure}[h]
    \centering\includegraphics[scale=0.5]{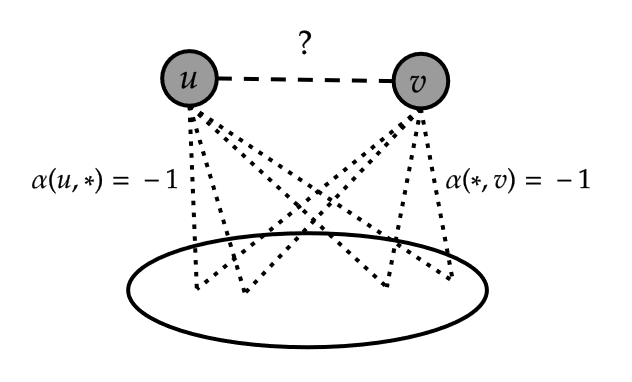}
    \caption{Diagram for proof of Lemma \ref{lemma:error-free-k-unknown}.}
    \label{fig:error-free-k-known}
\end{figure}
    
Consider the result of querying all but one pair $u,v$ of elements of the ground set $V$. Suppose all of the queries made involving $u$ and $v$ return $0$. Then it is impossible to determine whether $u$ and $v$ belong to the same cluster.
\end{proof}

\begin{proof}[\textbf{Proof of Lemma~\ref{lemma:information-theoretic-lowerbounds}.}]
 The crucial insight leveraged to establish this result is the following: existence of a non-adaptive algorithm which uses $q$ queries to learn a hidden partition on $n$ elements in the presence of $c\cdot q$  errors (when $k$ is unknown) implies the existence of a binary $(c\cdot q)$-error-correcting code of size $B_n$ with codewords of length $q$. For the case of $k$-known, a similar statement holds for a code of size ${n\brace k}$. We therefore obtain a lower bound on $q$ (in either $k$-regime) using known results from coding theory.

 Consider the case in which $k$ is unknown. In this setting, the relative distance of the implied code is $\delta(C) = 2c + 1/q \geq 2c$, and the rate is
\[
    R(C) = \frac{\log(B_n)}{q}
\]

In particular, by \eqref{eq:bound-on-rate-of-codes}, we have:
\[
    \frac{\log(B_n)}{q} \leq 1-H_2(c) + o(1).
\]
Hence, by \eqref{eq:asymptotics-of-bell-numbers}, we find that $q$ must satisfy
\[
    q = \Omega\left(\frac{n\log(n)}{1-H_2(c)}\right)
\]
in the $k$-unknown setting.

When $k$ is known, the size of the implied code becomes ${n\brace k}$, and employing a similar argument to the above, but using \eqref{eq:asymptotics-of-stirling-numbers} instead of \eqref{eq:asymptotics-of-bell-numbers} yields:
\[
    q = \Omega\left(\frac{n\log\left(k\right)}{1-H_2(c)}\right),
\]
as needed.
\end{proof}

\begin{proof}[\textbf{Proof of Lemma~\ref{lemma:RS-expected-queries}}.]
    We begin by establishing \eqref{eq:general-expected-number-of-queries-las-vegas}. We will assume without loss of generality that $V=[n]$. For any $i \in [n]$, let $a_i \in [k]$ be such that $i\in C_{a_i}$. For any $i\in [n]$ and any $a \in [k]\setminus\{a_i\}$, let $X_{i,a}$ be the indicator random variable of the event that element $i$ is compared to some element in cluster $C_a$ (before it gets compared with some element in the cluster $C_{a_{i}}$ and gets assigned to it).
    
    Note that, during the run of the algorithm, with the exception of a single representative from each cluster, each element $i$ gets compared to a unique element in its own cluster, amounting to $(n-k)$ many comparisons. The rest of the comparisons are inter-cluster comparisons, in which an element $i$ gets compared to a representative of a different cluster $C_a$. Note that each element gets compared to a representative of some cluster $C_a$ at most once.

    We also have:
    \[
        \E{}{X_{i,a}} = \Pr[X_{i,a} =1]  ={n_a \over n_{a_i} + n_{a}},
    \]
    since the above probability is equal to the probability that the first vertex in $C_a$ is processed before the first vertex in $C_{a_i}$ in a random permutation.
    
    We then have that the number of queries made by \texttt{Randomized\_RS} is given by:
    \begin{align*}
        \E{}{Q_{\cC}^{\texttt{RS}}}&= (n-k) + \sum_{\substack{i\in[n]\\ a\in[k]\setminus\{a_i\}}} \E{}{X_{i,a}} = (n-k) + \sum_{\substack{a,b \in[k] \\a \neq b}} {n_a n_b \over n_a + n_b}
    \end{align*}
    as needed.

    We now upper bound the above expression. For any $a,b\in [k]$, the inequality $4n_{a}n_{b}\leq (n_{a}+n_{b})^{2}$ holds; as a consequence, 
    \[
        \frac{n_{a}n_{b}}{n_{a}+n_{b}}\leq \frac{n_{a}+n_{b}}{4}.
    \]
    Therefore, 
    \begin{align*}
        \E{}{Q_{\cC}^{\texttt{RS}}}&= (n-k) + \sum_{\substack{a,b \in[k] \\a \neq b}} {n_a n_b \over n_a + n_b}\\
        &\leq (n-k) + \sum_{\substack{a,b \in[k] \\a \neq b}} {n_a+n_b \over 4}\\
        &= (n-k) + \sum_{a\in[k]}(k-1)\cdot \frac{n_{a}}{2}\\
        &= (n-k) + \frac{n\cdot(k-1)}{2}\\
        &= \frac{n\cdot(k+1)}{2} - k
    \end{align*}
\end{proof}

\begin{lemma}\label{lemma:repetition-bound}
    Consider the problem of learning a hidden partition on $n$ elements, given access to an $\ell$-faulty same-cluster oracle $\alpha$. Without loss of generality, we consider the problem where $k$ is unknown: the same statement and proof hold for $k$-known.
    
    Assume there exists an algorithm which, given a set $V$ of cardinality $n$ and a non-faulty same-cluster oracle, requires at most $q_{\text{error\_free}}$ queries in the worst case to learn a hidden partition on $V$. Then there exists an algorithm which learns a hidden partition on $V$ in the $\ell$-faulty regime which always makes at most $(\ell+1)q_{\text{error\_free}} +\ell$ queries in the worst case.
\end{lemma}
\begin{proof}
    Let \Algo$_{\operatorname{exact}}$ denote an algorithm which, given a set $V$ of cardinality $n$ and a non-faulty same-cluster oracle $\alpha_{\textrm{error-free}}$, requires at most $q_{\text{error\_free}}$ queries in the worst case to learn a hidden partition on $V$.

    Consider the following algorithm \Algo$_{\operatorname{faulty}}$ for the $\ell$-faulty version of the problem: \Algo$_{\operatorname{faulty}}$ simulates  \Algo$_{\operatorname{exact}}$, but every time a query is made by \Algo$_{\operatorname{exact}}$,  \Algo$_{\operatorname{faulty}}$ repeats that same query until it receives $\ell+1$ consistent answers, and then takes these consistent answers to be the answer to \Algo$_{\operatorname{exact}}$'s query. It is easy to see that \Algo$_{\operatorname{faulty}}$ correctly solves the problem with an $\ell$-faulty oracle.
    
    The following simple argument shows that \Algo$_{\operatorname{faulty}}$ makes at most $(\ell+1)q_{\text{error\_free}}+\ell$ queries in the worst case. We say a query made by \Algo$_{\operatorname{faulty}}$ is \emph{spurious} if its disagrees with $\ell+1$ other queries made to the same set of vertices. Since the oracle being queried is only $\ell$-faulty, every spurious query is an error, and the total number of spurious queries is at most $\ell$. On the other hand, by definition, \Algo$_{\operatorname{faulty}}$ makes at most $\ell+1$ non-spurious queries for each query made by \Algo$_{\operatorname{exact}}$. The bound then follows.

\end{proof}

\newpage
\section{Relevant background on R\'enyi-Ulam Games}\label{appendix:renyi-ulam}
In this section we provide some background on the history of the Ulam liar game defined in Section~\ref{subsection:background}, for the interested reader. We define the value of the game $L(N,\ell)$ as the minimax number of questions needed for the questioner to determine the chosen number with certainty. The game and its variations have been extensively studied and many different bounds have been proposed to the value of $L(N,\ell)$ \cite{pelc1987solution,czyzowicz1989ulam, pelc1989searching,deppe2000solution,pelc2002searching}. \cite{spencer1984guess} studies a version of this game in which $\ell=1$ and all the questions have to be of the form ``is $x < a$?'' for some $a \in N$. He then gives a questioner strategy for this version of the game that is guaranteed to find the correct value of $x$ after:
\[
    \lceil\log_2N\rceil + 2\sqrt{2 \lceil\log_2N\rceil} + 2
\]
questions. \cite{rivest1978coping} show a bound of 
\[
    \log_2N + \ell \log_2\log_2N + O(\ell \log \ell)
\]
for the same version of the problem. 

\cite{pelc1987solution} considers the more general query model, which admits questions of the form ``is $x \in T$'' for $T$ some arbitrary subset of $[N]$. His results establish that 
\[
    L(N,1) = \begin{cases}
    \argmin \{ k\in \N \mid N(k+1) \leq 2^k\} & \text{if }N\text{ is even,}\\
    \argmin \{ k\in \N \mid N(k+1) + (k-1)\leq 2^k\} & \text{if }N\text{ is odd.}
    \end{cases}
\]
In another paper, \cite{spencer1992ulam} introduces the framework of chip-liar games (See also~\cite{alon2016probabilistic}) to analyze $L(N,\ell)$. A consequence of this work is that:
\begin{equation}
    N \leq {2^{L(N,\ell)} \over \Vol(\ell,L(N,\ell))},
\end{equation}
where $\Vol(\ell,q)$  denotes the volume of the Hamming ball in $\{0,1\}^q$ of radius $\ell$:
\[
    \textrm{Vol}(\ell, q) \defeq \sum_{i=0}^\ell {n \choose i}.
\]
\cite{hill1998solution}, give the value of $L(2^{20},\ell)$ for every $\ell$, and in particular, they show that for every $\ell \geq 8$, it is true that $L(2^{20},\ell) = 2\ell+6$. 

Understanding variations of this game is still an active area of research (See, e.g.~\cite{marini2005probabilistic, cicalese2013variations,cicalese2014perfect,cicalese2020multi}) and they have recently found applications in the analysis of online algorithms with imperfect advice \cite{angelopoulos2023r}.

\newpage
\section{Uniquely $k$-colorable Graphs}\label{appendix:surjectively-uniquely-kc}
We define a $k$-\emph{coloring} of $G$ a map $\chi:V\to [k]$, such for all $u\neq v$ in $V$, $\chi(u)\neq \chi(v)$. In some literature, a map satisfying this property is referred to as a \emph{proper} $k$-coloring, but we will ignore this distinction. We may also omit the value $k$ in the term ``$k$-coloring'' and just write ``coloring'' when no ambiguity arises.

We say that two $k$-colorings $\chi_1 ,\chi_2$ are \emph{distinct} if for every bijective map $\sigma:[k] \to [k]$, there exists an element $v\in V$ such that $\chi_1(v) \neq \sigma \circ \chi_2(v)$. Note that this condition is equivalent to:
\[
    \{\chi_1^{-1}(v) \mid v\in V\} \neq \{\chi_2^{-1}(v) \mid v\in V\}.
\]
A graph $G$ is uniquely $k$-colorable if it is $k$-colorable, but there are no two distinct $k$-colorings of $G$. Shaoji \cite{shaoji1990size} showed that every uniquely $k$-colorable graph $G = (V,E)$ with $|E|=m$ and $|V|=n$ satisfies:
\begin{equation}\label{eq:Shaojis-bound}
    m \geq n(k-1) - {k \choose 2}.
\end{equation}

This paper focuses on a slightly different property which we now define.

\begin{definition}[Uniquely Surjectively $k$-colorable Graph]
We say a graph is \emph{uniquely surjectively $k$-colorable} if there exists a \textbf{surjective} map $\chi:V \to [k]$ such that, for all $uv\in E$: $\chi(v) \neq \chi(u)$, and such that for any two such maps $\chi_{1}$ and $\chi_{2}$, there exists a bijective function $\sigma:[k]\to[k]$ such that $\chi_{2}=\sigma\circ\chi_{1}$. This simply means that there is a unique way to color the graph $G$ using \textbf{exactly} $k$ \textbf{distinct} colors. 
\end{definition}

We show the following:
\begin{proposition}\label{prop:unique-colorings}
Let $G=(V,E)$ be a graph with $|V| = n$, $|E|=m$ and let $k$ be positive integer. We have:
\begin{enumerate}
    \item If $k<n$, then $G$ is uniquely $k$-colorable if and only if $G$ is uniquely surjectively $k$-colorable,
    \item If $k=n$, then $G$ is always uniquely surjectively $k$-colorable,
    \item If $k=n$ then $G$ is uniquely $k$-colorable if and only if $G$ is the complete graph $K_n$.
\end{enumerate}
\end{proposition}
\begin{proof}

We begin by proving both directions of part 1.

($\Rightarrow$) Let $G$ be uniquely $k$-colorable. We begin by showing that $G$ admits a coloring that uses exactly $k$-colors. Since $G$ is $k$-colorable then $G$ admits a coloring $\chi$ which uses at most $k$ colors. If $|\chi(V)| =k$ then we are done. Otherwise, since $k < n$, there exists a pair of vertices that are colored with the same color as each other. One can then pick one of those vertices and color it with a color $i \not\in \chi(V)$, producing a new proper coloring $\chi'$ which uses strictly more colors than $\chi$. It is easy to see that $\chi'$ satisfies $\chi'(u) \neq \chi'(v)$ for every $uv\in E$, since $\chi$ satisfied the same condition. Repeating this procedure yields the desired surjective $k$-coloring. Since there is a unique $k$-coloring of $G$ and $G$ admits a $k$-coloring with exactly $k$ colors, then there is a unique surjective $k$-coloring of $G$, so $G$ is uniquely surjectively $k$-colorable.

($\Leftarrow$) Let $k <n$ and let $G$ be uniquely surjectively $k$-colorable. Then $G$ is clearly $k$-colorable with some surjective coloring $\chi$, moreover, no other surjective $k$-coloring exists, so any other coloring of $G$ must use strictly fewer than $k$ colors. We now show that no such coloring can exist. 

Suppose, for the sake of contradiction, there exists a coloring $\chi_2$ using $k'<k$ colors. By the procedure in the previous paragraph, one can construct a new coloring $\chi_3$ which uses exactly $k-1$ colors. Let $C_1, ... ,C_{k-1}$ be its corresponding color classes. Since $k<n$ at least one of the following must hold: either there exists $C_i$ satisfying $|C_i| \geq 3$ or there exists two color classes $C_i$ and $C_j$ such that $|C_i|\geq 2$ and $|C_j|\geq 2$. In each case we can produce $2$ distinct surjective $k$-colorings of $G$. In the first case, we do so by breaking up the color class $C_i$ in two different ways: fix two distinct elements $u,v$ in $C_i$. Let $\chi_u$ and $\chi_v$ be the colorings given by:
\[
    \chi_u(x) = \begin{cases} \chi_3(x) &x\in V\setminus\{u\},\\  c^* & x=u, \end{cases} \hspace{0.5cm}\text{and}\hspace{0.5cm} \chi_v(x) = \begin{cases} \chi_3(x) &x\in V\setminus\{v\},\\  c^* & x=v, \end{cases} 
\]
where $c^{*}$ is an element in $[k] \setminus \chi_3(V)$. This way we obtain two distinct proper surjective $k$-colorings of $G$, $\chi_u$ and $\chi_v$. Now, consider the case in which $|C_i| \geq 2$ and $|C_j|\geq 2$, let $u \in C_i$ and $v\in C_j$. In this case too, $\chi_v$ and $\chi_u$ as defined above are two distinct surjective $k$-colorings of $G$. So two distinct surjective $k$-colorings of $G$ must exist. This leads to a contradiction, since $G$ was assumed to be uniquely surjectively $k$-colorable. This completes the proof of part 1.

For part 2, we note that any graph admits the surjective coloring which assigns a distinct color to every vertex. Furthermore, since $n=k$ there is no other way to surjectively color the vertices of $G$ with $n$ colors (once we account for color permutations), yielding the statement.

For part 3, let $G$ be any graph other than $K_n$. There exists two distinct vertices $u,v\in V$ such that $uv \not\in E$. Then $G$ admits two distinct $n$-colorings, $\chi$ and $\chi'$ defined as follows: $\chi$ assigns all vertices to a distinct color. $\chi'$ assigns all the vertices to a distinct color, with the exception of $u$ and $v$ which are assigned the same color. Hence $G$ cannot be uniquely $k$-colorable. On the other hand, it is easy to see that $K_n$ is uniquely $n$-colorable, since the coloring assigning each vertex to a distinct color is the only valid $n$-coloring of $K_n$, completing the proof of part 3.
\end{proof}
The following simple consequence of~\eqref{eq:Shaojis-bound} and Proposition \ref{prop:unique-colorings} is a central ingredient in the body of the paper.
\begin{lemma}
Let $k$ be a positive integer. Let $G=(V,E)$ be a uniquely surjectively $k$-colorable graph with $|V| =n$ and $|E|=m$, if $k<n$, then:
\[
     m \geq n(k-1) - {k \choose 2}.
\]
\end{lemma}
\citet{LM2021} make implicit use of this result, since the query complexity bounds obtained in their work would only be applicable to the case $k<n$. Hence, this appendix highlights the result in Proposition \ref{prop:unique-colorings}, in order to clarify that even if our work studies the query complexity question on all pairs $(k,n)$, the behaviour of such complexity bounds differ in the case $k=n$.

\newpage
\section{Algorithms}\label{appendix:algorithms}
\SetKwRepeat{Do}{do}{while}
In this appendix, we give detailed descriptions of the algorithms used in the body of the paper, for reference purposes. We begin with the vanilla version of the algorithm of \cite{ReyzinSrivastava2007} for learning partitions given access to an (error-free) same-cluster oracle. This was proved to achieve the optimal query complexity for the problem by \cite{LM2021}. The first version is used when $k$ is unknown to the learner.

\begin{algorithm2e}[h]
\caption{\texttt{RS}$(V,\alpha)$}\label{algo:vanilla-RS-k-unknown}
    \KwIn{Finite set $V$, (error-free) same-cluster oracle $\alpha$ for some partition $\cC$ of $V$}
    \KwOut{$\cC \defeq \{\cC_i\}_{i=1}^k$ a partition of $V$.}
    $num\_clusters=0$\\
    // Each iteration of the following loop processes a single vertex\\
    \While{$V\neq \emptyset$} 
    { 
        $Is\_Placed = $ False\\
        Pick $v \in V$ arbitrarily.\\
        \For{$a=1, ... , num\_clusters$} {
            Query $\alpha(v,w)$ for some $w\in C_a$\\
            \If{$\alpha(v,w) = 1$} {
                $C_a = C_a \cup\{v\}$\\
                $Is\_Placed = $ True\\
                \textbf{break}
            }
        }
        \If {$Is\_Placed ==$ False}{
            $num\_clusters = num\_clusters+1$\\
            $C_{num\_clusters} = \{v\}$
        }
        $V = V \setminus \{v\}$\\
    }
\Return{$\cC\defeq \{C_a\}_{a=1}^{k}$}\;
\end{algorithm2e}

\newpage
A similar algorithm is used when $k$ is known to the learner (Algorithm~\ref{algo:vanilla-RS-k-known}). This algorithm is analogous to the previous one, but if one fails to place an element $v$ of $V$ in $k-1$ distinct clusters, then it uses the fact that $k$ is known to infer that $v$ must be in the remaining cluster, thus saving $(n-k)$ queries in the worst case.
\begin{algorithm2e}[h]
\caption{\texttt{RS\_k}$(V,\alpha,k)$}\label{algo:vanilla-RS-k-known}
    \KwIn{Finite set $V$, same-cluster oracle $\alpha$, number of clusters $k$}
    \KwOut{$\cC \defeq \{\cC_i\}_{i=1}^k$ a partition of $V$.}
    $num\_clusters=0$\\
    // Each iteration of the following loop processes a single vertex\\
    \While{$V \neq \emptyset$} 
    { 
        Pick $v \in V$ arbitrarily.\\
        $Is\_Placed = $ False\\
        \For{$a=1, ... , \min\{num\_clusters,k-1\}$} {
            Query $\alpha(v,w)$ for some $w\in C_a$\\
            \If{$\alpha(v,w) = 1$} {
                $C_a = C_a \cup\{v\}$\\
                $Is\_Placed = $ True\\
                \textbf{break}
            }
        }
        \If {$Is\_Placed ==$ False}{
            \If {$num\_clusters <k$} {
            $num\_clusters = num\_clusters+1$\\
            $C_{num\_clusters} = \{v\}$
            }
            \Else {
                $C_{num\_clusters} = C_{num\_clusters}\cup \{v\}$
            }
        }
        $V = V \setminus \{v\}$\\
    }
\end{algorithm2e}
\newpage
We then give a randomized algorithm for the error-tolerant version of the problem. This algorithm's worst-case and expected performance are analyzed in Section~\ref{ssec:adaptive-error-tolerant-upperbound} and Section~\ref{section:las-vegas-type-bounds} respectively.
\begin{algorithm2e}[h]
\caption{\texttt{Robust\_RS}$(V,\alpha,\ell)$}\label{algo:Robust-RS}
    \KwIn{A finite set $V$, same-cluster oracle $\alpha$, error-tolerance $\ell$.}
    \KwOut{$\cC \defeq \{\cC_i\}_{i=1}^k$ a partition of $V$.}
    $num\_cluster=0$\\
    // Each iteration of the following loop processes a single vertex\\
    \While{$V\neq \emptyset$} 
    { 
        $Is\_Placed = $False\\
        Pick $v \in V$ uniformly at random.\\
        \For{$a=1, ... , num\_clusters$} {
            For some $w\in C_a$, query $\alpha(v,w)$ repeatedly until $\ell+1$ of the responses received agree. Let $r \in \{\pm 1\}$ be the value of those responses. \\
            \If{$r = 1$} {
                $C_a = C_a \cup\{v\}$\\
                $Is\_Placed$ = True\\
                \textbf{break}
            }
        }
        \If{$Is\_Placed == $\emph{False}}{
            $num\_cluster = num\_clusters+1$\\
            $C_{num\_clusters} = \{v\}$
        }
        $V = V \setminus \{v\}$\\
    }
\end{algorithm2e}

\begin{algorithm2e}[h]
\caption{\texttt{Robust\_RS\_k}$(V,\alpha,k,\ell)$}\label{algo:Robust-RS-k}
    \KwIn{A finite set $V$, same-cluster oracle $\alpha$, number of clusters $k$, error-tolerance $\ell$.}
    \KwOut{$\cC \defeq \{\cC_i\}_{i=1}^k$ a partition of $V$.}
    $num\_cluster=0$\\
    // Each iteration of the following loop processes a single vertex\\
    \While{$V\neq \emptyset$} 
    { 
        $Is\_Placed = $False\\
        Pick $v \in V$ uniformly at random.\\
        \For{$a=1, ... ,\min \{num\_clusters, k-1\}$} {
            For some $w\in C_a$, query $\alpha(v,w)$ repeatedly until $\ell+1$ of the responses received agree. Let $r \in \{\pm 1\}$ be the value of those responses. \\
            \If{$r = 1$} {
                $C_a = C_a \cup\{v\}$\\
                $Is\_Placed$ = True\\
                \textbf{break}
            }
        }
        \If{$Is\_Placed == $\emph{False}}{
            \If {$num\_clusters <k$} {
            $num\_cluster = num\_clusters+1$\\
            $C_{num\_clusters} = \{v\}$
            }
            \Else {
                $C_{num\_clusters} = C_{num\_clusters}\cup \{v\}$
            }
        }
        $V = V \setminus \{v\}$\\
    }
\end{algorithm2e}

\newpage
\noindent We then give the following parallelized version of the algorithms of Reyzin and Srivastava, which are used to show that these algorithms can be made to run in at most $k$ rounds of adaptivity. As per the non-parallel versions, there is an algorithm for the case of $k$ unknown (Algorithm~\ref{algo:parallel-learn}) and one for the case of $k$ known (Algorithm~\ref{algo:parallel-learn-k}).
\begin{algorithm2e}[h]
\caption{\texttt{Parallel\_RS}$(V,\alpha)$}\label{algo:parallel-learn}
    \KwIn{A finite set $V$, same-cluster oracle $\alpha$.}
    \KwOut{$\cC \defeq \{\cC_i\}_{i=1}^k$ a partition of $V$.}
    $num\_clusters = 0$\\
    // All the queries made in a single iteration of the while loop are submitted at once
    \While{$V \neq \emptyset$} {
        $num\_clusters =num\_clusters+1$\\
        Fix some $v \in V$, then query $\alpha(u,v)$ for every $u\in V\setminus\{v\}$\\
        $C_{num\_clusters} = \{v\}\cup \{u\in V\setminus\{v\} \mid \alpha(u,v) = 1\}$\\
        $V = V \setminus C_{num\_clusters}$
    }
\Return{$\cC\defeq \{C_a\}_{a=1}^{k}$}\;
\end{algorithm2e}

\begin{algorithm2e}[h]
\caption{\texttt{Parallel\_RS\_k}$(V,\alpha, k)$}\label{algo:parallel-learn-k}
    \KwIn{A finite set $V$, a same-cluster oracle $\alpha$, a number of clusters $k \in \mathbb{N}$.}
    \KwOut{$\cC \defeq \{\cC_i\}_{i=1}^k$ a partition of $V$.}
    // All the queries made in a single iteration of the while loop are submitted at once
    \For{$a=1,\dots,k-1$} {
        $k = k+1$\\
        Fix some $v \in V$, then query $\alpha(u,v)$ for every $u\in V\setminus\{v\}$\\
        $C_a = \{v\}\cup \{u\in V\setminus\{v\} \mid \alpha(u,v) = 1\}$\\
        $V = V \setminus C_a$
    }
    $C_k = V$\\
\Return{$\cC\defeq \{C_a\}_{a=1}^{k}$}\;
\end{algorithm2e}

\newpage
We now give randomized versions of the algorithms of Reyzin and Srivastava, which we used to obtain Las-Vegas bounds on the runtime of our error-tolerant algorithm. As before, we have a $k$-unknown version (Algorithm~\ref{algo:randomized-k-unknown}) and a $k$-known version (Algorithm~\ref{algo:randomized-k-known}).
\begin{algorithm2e}[h]
\caption{\texttt{Randomized\_RS}$(V,\alpha)$}\label{algo:randomized-k-unknown}
    \KwIn{A finite set $V$, same-cluster oracle $\alpha$}
    \KwOut{$\cC \defeq \{\cC_i\}_{i=1}^k$ a partition of $V$.}
    $num\_clusters=0$\\
    // Each iteration of the following loop processes a single vertex\\
    \While{$V\neq \emptyset$} 
    { 
        $Is\_Placed =$ False 
        Pick $v \in V$ uniformly at random.\\
        \For{$a=1, ... , num\_clusters$} {
            Query $\alpha(v,w)$ for some $w\in C_a$\\
            \If{$\alpha(v,w) = 1$} {
                $C_a = C_a \cup\{v\}$\\
                $Is\_Placed =$ True
                \textbf{break}
            }
        }
        \If {$Is\_Placed ==$ False}{
            $num\_clusters = num\_clusters+1$\\
            $C_{num\_clusters} = \{v\}$
        }
        $V = V \setminus \{v\}$\\
    }
\Return{$\cC\defeq \{C_a\}_{a=1}^{k}$}\;
\end{algorithm2e}

\begin{algorithm2e}[h]
\caption{\texttt{Randomized\_RS\_k}$(V,\alpha,k)$}\label{algo:randomized-k-known}
    \KwIn{A finite set $V$, same-cluster oracle $\alpha$, a number of clusters $k\in \N$.}
    \KwOut{$\cC \defeq \{\cC_i\}_{i=1}^k$ a partition of $V$.}
    $num\_clusters=0$\\
    // Each iteration of the following loop processes a single vertex\\
    \While{$V\neq \emptyset$} 
    { 
        $Is\_Placed =$ False\\ 
        Pick $v \in V$ uniformly at random.\\
        \For{$a=1, ... , \min\{num\_clusters,k-1\}$} {
            Query $\alpha(v,w)$ for some $w\in C_a$\\
            \If{$\alpha(v,w) = 1$} {
                $C_a = C_a \cup\{v\}$\\
                $Is\_Placed =$ True 
                \textbf{break}
            }
        }
        \If {$Is\_Placed ==$ False }{
            \If {$num\_clusters <k$} {
            $num\_clusters = num\_clusters+1$\\
            $C_{num\_clusters} = \{v\}$
            }
            \Else {
                $C_{num\_clusters} = C_{num\_clusters} \cup \{v\}$
            }
        }
        $V = V \setminus \{v\}$\\
    }
\Return{$\cC\defeq \{C_a\}_{a=1}^{k}$}\;
\end{algorithm2e}
\end{document}